\documentclass[11pt]{article}
\usepackage[english]{babel}
\usepackage{cite}
\usepackage{amssymb,amsmath,amsthm,mathrsfs}

\newtheorem{lemma}{Lemma}
\newtheorem{proposition}{Proposition}
\newtheorem{theorem}{Theorem}
\newtheorem{corollary}{Corollary}

\theoremstyle{definition}
\newtheorem{definition}{Definition}

\theoremstyle{remark}
\newtheorem{remark}{Remark}

\newcommand{\fin}{\hspace{0.5mm}}

\newcommand{\spa}{\hspace{-2mm}}

\newcommand{\tre}{\hspace{0.3mm}}
\newcommand{\quattro}{\hspace{0.4mm}}
\newcommand{\cinque}{\hspace{0.5mm}}
\newcommand{\sei}{\hspace{0.6mm}}

\newcommand{\otto}{\hspace{0.8mm}}

\newcommand{\mtre}{\hspace{-0.3mm}}

\newcommand{\mcinque}{\hspace{-0.5mm}}

\newcommand{\ima}{\mathrm{i}}
\newcommand{\siy}{\sigma_y}
\newcommand{\trasp}{{\hspace{-0.2mm}\mbox{\tiny $\mathsf{T}$}}}

\newcommand{\defi}{\mathrel{\mathop:}=}
\newcommand{\ifed}{=\mathrel{\mathop:}}

\newcommand{\tr}{\mathrm{tr}}

\newcommand{\hh}{\mathcal{H}}
\newcommand{\jj}{\mathcal{J}}
\newcommand{\kk}{\mathcal{K}}
\newcommand{\pkk}{P_{\kk}}

\newcommand{\ran}{\mathrm{ran}\hspace{0.3mm}}
\newcommand{\rank}{\mathrm{rank}}

\newcommand{\imp}{\check{P}}
\newcommand{\imq}{\check{Q}}
\newcommand{\impkk}{\check{P}_{\kk}}
\newcommand{\prkk}{\mathscr{P}_{\hspace{-0.2mm}k}\hspace{-0.2mm}(\kk)}

\newcommand{\imprkk}{\check{\mathscr{P}}_{\hspace{-0.2mm}k}\hspace{-0.2mm}(\kk)}
\newcommand{\imkk}{\check{\kk}}
\newcommand{\primkk}{\mathscr{P}_{\hspace{-0.2mm}k}\hspace{-0.2mm}(\imkk)}

\newcommand{\pimkk}{P_{\imkk}}

\newcommand{\neop}{\widetilde{P}}
\newcommand{\neopsq}{\widetilde{P}^{\hspace{0.2mm}2}}
\newcommand{\neoq}{\widetilde{Q}}

\newcommand{\bp}{\breve{P}}
\newcommand{\bq}{\breve{Q}}

\newcommand{\pper}{P^\perp}

\newcommand{\neopper}{\neop^\perp}
\newcommand{\neoqper}{\neoq^\perp}

\newcommand{\bimkk}{\mathscr{B}(\imkk)_{\mathsf{s}}}
\newcommand{\bkk}{\mathscr{B}(\kk)_{\mathsf{s}}}
\newcommand{\imbkk}{\check{\mathscr{B}}(\kk)_{\mathsf{s}}}

\newcommand{\nat}{\mathbb{N}}

\newcommand{\id}{I}
\newcommand{\ide}{\mathrm{Id}}

\newcommand{\mms}{\rho_{\star}}

\newcommand{\dom}{\mathrm{dom}\hspace{0.2mm}}

\newcommand{\spanr}{\mathrm{span}_{\hspace{0.2mm}\mathbb{R}}}
\newcommand{\spancl}{\overline{\mathrm{span}}_{\hspace{0.2mm}\mathbb{R}}}

\newcommand{\prk}{\mathscr{P}_{\hspace{-0.2mm}k}\hspace{-0.2mm}(\hh)}
\newcommand{\prdk}{\mathscr{P}_{\hspace{-0.2mm}2k}\hspace{-0.2mm}(\hh)}
\newcommand{\prm}{\mathscr{P}_{\hspace{-0.2mm}m}\hspace{-0.2mm}(\hh)}
\newcommand{\pru}{\mathscr{P}_1\hspace{-0.2mm}(\hh)}
\newcommand{\prnm}{\mathscr{P}_{\hspace{-0.4mm}\frac{n}{2}}\hspace{-0.4mm}(\hh)}

\newcommand{\fr}{\mathscr{F}(\hh)}
\newcommand{\frsa}{\mathscr{F}(\hh)_{\mathsf{s}}}

\newcommand{\stah}{\mathscr{S}(\hh)}
\newcommand{\frstah}{\mathscr{F}\hspace{-1mm}\mathscr{S}(\hh)}
\newcommand{\staj}{\mathscr{S}(\jj)}
\newcommand{\stauk}{\mathscr{S}_{\hspace{-0.2mm}k}\hspace{-0.2mm}(\hh)_{\mathsf{u}}}
\newcommand{\stauu}{\mathscr{S}_{\hspace{-0.2mm}1}\hspace{-0.2mm}(\hh)_{\mathsf{u}}}
\newcommand{\staun}{\mathscr{S}_{\hspace{-0.2mm}n}\hspace{-0.2mm}(\hh)_{\mathsf{u}}}
\newcommand{\trc}{\mathscr{B}_1\hspace{-0.2mm}(\hh)}
\newcommand{\trcsa}{\trc_{\mathsf{s}}}
\newcommand{\bou}{\mathscr{B}(\hh)}
\newcommand{\bousa}{\mathscr{B}(\hh)_{\mathsf{s}}}
\newcommand{\poscon}{\trc^{\hspace{-0.2mm}\mbox{\tiny $+$}}\hspace{-0.3mm}}
\newcommand{\poco}{\trc_{\hspace{-0.1mm}\ast}^{\hspace{-0.2mm}\mbox{\tiny $+$}}\hspace{-0.3mm}}
\newcommand{\neco}{\trc_{\hspace{-0.1mm}\ast}^{\hspace{-0.2mm}\mbox{\tiny $-$}}\hspace{-0.3mm}}
\newcommand{\sico}{\trc_{\hspace{-0.1mm}\ast}^{\hspace{-0.2mm}\mbox{\tiny $\pm$}}\hspace{-0.3mm}}
\newcommand{\prom}{\varpi}

\newcommand{\lima}{\Phi}
\newcommand{\domlima}{\dom(\lima)}
\newcommand{\limacl}{\widehat{\lima}}
\newcommand{\limb}{\Psi}
\newcommand{\limr}{\lima_0}

\newcommand{\linp}{\Lambda^{\hspace{-0.5mm}\mbox{\tiny $[k]$}}\hspace{-0.2mm}}
\newcommand{\linpone}{\Lambda^{\hspace{-0.5mm}\mbox{\tiny $[1]$}}\hspace{-0.2mm}}
\newcommand{\linpu}{\Lambda_{U}^{\hspace{-0.5mm}\mbox{\tiny $[k]$}}\hspace{-0.2mm}}

\newcommand{\map}{\phi}
\newcommand{\prelima}{\Theta}
\newcommand{\prelimmo}{\Theta^{-1}}
\newcommand{\limco}{\lima_{\mathbb{C}}}

\newcommand{\limcocl}{\widehat{\lima}_{\mathbb{C}}}

\newcommand{\syw}{\mathcal{W}}
\newcommand{\psyw}{\breve{\syw}}

\begin{document}

\title{Symmetry witnesses}

\author{
Paolo Aniello$^{1,2}$ and Dariusz Chru\'sci\'nski$^3$
\vspace{2mm}
\\ \small \it
$^1$Dipartimento di Fisica ``Ettore Pancini'', Universit\`a di Napoli ``Federico II'',
\\ \small \it
Complesso Universitario di Monte S.~Angelo, via Cintia, I-80126 Napoli, Italy
\vspace{2mm}
\\ \small \it
$^2$Istituto Nazionale di Fisica Nucleare, Sezione di Napoli,
\\ \small \it
Complesso Universitario di Monte S.~Angelo, via Cintia, I-80126 Napoli, Italy
\vspace{2mm}
\\ \small \it
$^3$Institute of Physics, Faculty of Physics, Astronomy and Informatics,
\\ \small \it
Nicolaus Copernicus University,
\\ \small \it
Grudziadzka 5, 87–100 Toru\'n, Poland
}

\date{}

\maketitle

\begin{abstract}
\noindent A \emph{symmetry witness} is a suitable subset of the space of selfadjoint
trace class operators that allows one to determine whether a linear map is a symmetry
transformation, in the sense of Wigner. More precisely, such a set is invariant with respect to
an injective densely defined linear operator in the Banach space of selfadjoint
trace class operators (if and) only if this operator is a symmetry transformation.
According to a linear version of Wigner's theorem, the set of pure states --- the rank-one projections ---
is a symmetry witness. We show that an analogous result holds for the set of projections
with a fixed rank (with some mild constraint on this rank, in the finite-dimensional case).
It turns out that this result provides a complete classification of the sets of projections with
a fixed rank that are symmetry witnesses. These particular symmetry witnesses are \emph{projectable};
i.e., reasoning in terms of quantum states, the sets of `uniform' density operators of corresponding
fixed rank are symmetry witnesses too.
\end{abstract}

%%%------------------------------------------------------------------------------
\section{Introduction}
\label{intro}
%%%------------------------------------------------------------------------------

Wigner's theorem~\cite{Wigner,Bargmann,Kadison,Simon,Ludwig,Cassinelli,Cassinelli-book,Parthasarathy,Grabowski}
is one of the pillars of quantum mechanics. It expresses with adamantine clarity the fact
the symmetry transformations of quantum mechanics are associated, in a natural way,
with the elements of the unitary-antiunitary group of the relevant Hilbert space, modulo phase factors.
The theorem can be derived in several ways and formulated in several essentially equivalent forms --- see,
e.g.,~\cite{Simon,Ludwig,Cassinelli,Cassinelli-book,Parthasarathy,Grabowski}, and references therein ---
and it has given rise to an evergreen area of research between theoretical physics and mathematics~\cite{Li,Molnar}.
The literature on this topic is so vast that it would be futile to try to give an account of it here
(in addition to the previous references, we will content ourselves with resuming a few salient facts
in sect.~\ref{background}).

In our present contribution, we will look at Wigner's theorem from a perspective close to
the theory of open quantum systems and to quantum information science~\cite{Davies-book,Breuer,Bengtsson,Nielsen},
where one typically considers both the pure and the mixed states of a quantum system.
In this context, one very often deals with \emph{linear} maps in the space of trace class operators
(the linear space where the quantum states live) that are positive and trace-preserving.
The symmetry transformations, in the sense of Wigner,
are obviously maps of this kind. Notice, however, that symmetries in Wigner's seminal work
were originally regarded as maps on \emph{pure} states only. Linearity, or convex linearity,
appear in subsequent formulations involving maps defined on more general structures
(e.g., the so-called Kadison and Jordan-Segal automorphisms, see~\cite{Simon}).
It turns out that linearity is quite a strong assumption, so that one of the main hypotheses
of the classical formulation of Wigner --- preservation of the transition probabilities ---
can be skipped in the linear setting.

Recently, an interesting partial generalization of the linear version of Wigner's theorem
has been obtained~\cite{Sarbicki}. It stems from a very natural problem: Given a linear map acting in the
space of (finite-dimensional) hermitian matrices, which preserves the set of projections of a fixed rank,
for what values of this rank one can conclude that the map is a symmetry transformation?
Introducing a notion of \emph{symmetry witness} --- a suitable subset of the space of hermitian matrices or,
more generally, of the space of selfadjoint trace class operators --- the given problem
can be reformulated as follows: \emph{When the set of projections of a fixed rank is a symmetry witness?}
By the linear version of Wigner's theorem (see Theorem~\ref{linwig} \emph{infra})
the set of pure states --- the rank-one projections --- is a symmetry witness.

Methods developed in~\cite{Sarbicki} provide a complete solution of the problem
in the case where the Hilbert space dimension is a prime number, but with some further insights
that turn out to be essential for the solution of the general problem; solution that is
the main result of our paper.

With the previously outlined picture in mind, we will consider suitable linear operators
in the Banach space of selfadjoint trace class operators that map the set of
projections of a given rank onto itself. Wigner's theorem deals with
the most relevant case: the rank-one projections are preserved. We do not put restrictions on the dimension of
the carrier Hilbert space (we just neglect the trivial one-dimensional case),
and we do not assume that our maps be positive and/or trace-preserving; thus we are
in the same setting of the linear version of Wigner's theorem.
We will prove that in the infinite-dimensional case preservation of the projections
of an arbitrarily fixed rank leads to a symmetry transformation. In the finite-dimensional
case, the only natural constraint --- the relevant value of the rank is smaller than
the Hilbert space dimension $n$ --- is assumed if $n$ is \emph{odd} or $n=2$; for $n\ge 4$ and \emph{even}, the value
$n/2$ must be also discarded. A simple counter-example shows why. This result gives
a complete classifications of the sets of projections of a fixed rank that are symmetry
witnesses. Reasoning in terms of quantum states, this result amounts to providing a
complete classification of the sets of `uniform' density operators of a fixed (finite) rank
that are symmetry witnesses.

The paper is organized as follows. In sect.~\ref{background}, we will fix the notations
and then outline the conceptual framework underlying our work, establishing some basic facts
and providing a formal definition of a symmetry witness. Some useful technical facts will be proved in sect.~\ref{interlude}.
A reader who is not interested in these technicalities may go directly to sect.~\ref{main},
where we will derive our main results. Finally, in sect.~\ref{conclusions}, conclusions will be drawn.

%%%------------------------------------------------------------------------------
\section{Background and notations}
\label{background}
%%%------------------------------------------------------------------------------

In this section, we will fix the main assumptions and notations, and we will briefly describe
the general conceptual background of our work. In particular, we will introduce a
notion of \emph{symmetry witness}.

In the following, we will consider a separable complex Hilbert space $\hh$, with $\dim(\hh)\ge 2$
(apart from neglecting the trivial case $\dim(\hh)=1$, we do not put restrictions on $\dim(\hh)$),
and the real Banach space $\trcsa$ of selfadjoint trace class operators in $\hh$.
The convex body of density operators in $\hh$  --- the unit trace, positive trace class operators ---
will be denoted by $\stah\subset\trcsa$. The set of all (selfadjoint) projections of rank $k\in\nat$ will be denoted
by $\prk$, and hence, in particular, $\pru\subset\stah$ will indicate the set of \emph{pure} states. The symbol
$\frsa$ will denote the (dense) linear manifold of all selfadjoint operators of finite rank;
i.e., $\frsa\defi\spanr(\pru)$.

The (unique) canonical spectral decomposition of an operator in $\trcsa$, which converges
in the trace norm, will often be regarded as a (in general, non-unique) real linear combination
of mutually orthogonal rank-one projections, again converging in the trace norm. Clearly, this
just amounts to decomposing the spectral projections of a selfadjoint trace class operator into
mutually orthogonal minimal projections.

We will also consider the sets
\begin{equation}
\poco\defi\{A\in\trc\colon A\ge 0, \ A\neq 0\}\fin , \ \ \neco\defi-\poco\fin ,
\end{equation}
\begin{equation}
\sico\defi\poco\cup\neco\fin ,
\end{equation}
and we will denote by $\prom$ the natural projection map of $\sico$ onto $\stah$:
\begin{equation}
\prom\colon\sico\ni A \mapsto \tr(|A|)^{-1}\tre |A|\in\stah \fin .
\end{equation}
The set $\stauk\defi\prom(\prk)$ consists of all \emph{uniform} density operators
of rank $k$; namely, of all mixed states whose eigenvalues form a uniform (i.e., constant)
probability distribution; in particular, $\stauu\equiv\pru$ and, if $n=\dim(\hh)<\infty$,
$\staun=\{\mms\}$, where $\mms\defi n^{-1}\tre \id$ is the maximally mixed state.

In the finite-dimensional case, clearly $\frsa=\trcsa=\bousa$ (regarded as linear spaces),
and we find it convenient to use the `neutral' symbol $\bousa$. In this case, given a
projection $P$, we will denote by $\pper$ the projection onto the orthogonal complement
of the subspace $\ran(P)$ of $\hh$; i.e., $P\mapsto\pper\defi\id - P$ is the orthocomplementation map.

The celebrated theorem of Wigner, in its elegant original formulation, states that
a bijective map of $\pru$ onto itself, preserving the transition probabilities
--- $\tr(P\tre Q)$, $P,Q\in\pru$ ---
has the well known canonical form of a symmetry transformation associated with
an operator in $\hh$, which is either unitary or
antiunitary~\cite{Wigner,Bargmann,Simon,Ludwig,Cassinelli,Cassinelli-book},
and unique modulo phase factors.

As previously mentioned, the theorem has been reformulated giving rise to several
alternative versions, according to various possible purposes and applications.
We will consider here a \emph{linear} --- namely, involving linear operators
in the Banach space $\trcsa$ --- version of the theorem.

\begin{theorem}[Wigner] \label{linwig}
Let $\lima$ be a linear operator in $\trcsa$, defined on the dense domain $\dom(\lima)=\frsa=\spanr(\pru)$,
and mapping the set $\pru$ of pure states onto itself.
Then, $\lima$ is closable, and its closure $\limacl$ is a surjective isometry of the form
\begin{equation} \label{canfor}
\limacl(A)= U A \sei U^\ast, \ \ \ \forall\cinque A\in\trcsa \fin ,
\end{equation}
where $U$ is a unitary or antiunitary operator in $\hh$, uniquely defined up to
multiplication by a phase factor.
\end{theorem}

\begin{proof}
By the spectral decomposition of a selfadjoint finite rank operator
(regarding each spectral projection as a sum of mutually orthogonal rank-one projections),
and exploiting the fact that $\pru$ is mapped by $\lima$ into itself,
one finds out that $\lima$ is positive, trace-preserving and bounded on its dense domain.
Hence, $\lima$ is closable and its closure $\limacl$ is a bounded operator defined on
$\dom(\limacl)=\trcsa$. Notice that, by the spectral decomposition of a selfadjoint trace class operator,
which converges in the trace norm, one concludes as above that $\limacl$ is a trace-preserving and positive too.
Moreover, the spectral decomposition also implies that $\ran(\limacl)=\trcsa$,
because $\limacl$ is bounded and $\limacl(\pru)=\pru$.

The fact that $\limacl$ is a surjective positive linear map,
mapping the set of all \emph{pure elements} of the positive cone $\poscon$
of $\trcsa$ (i.e., the rank-one positive operators~\cite{Davies-book})
into itself --- namely, a surjective \emph{pure} positive map ---
entails, by a classical result of Davies (Theorem~{3.1} of~\cite{Davies};
also see Theorem~{3.1}, in chapt.~{2} of~\cite{Davies-book}),
that $\limacl$ is of the form $A\mapsto U A \sei U^\ast$,
where $U$ is a bounded linear or antilinear operator in $\hh$. This
operator must be, in particular, either unitary or antiunitary.
Indeed, by assumption $\limacl(\pru)=(\pru)$, and it is
easy to see that this condition implies that $U$ maps the unit sphere of $\hh$ onto itself.
(One arrives to the same conclusion exploiting the fact that $\limacl$ is trace-preserving
and surjective.) Clearly, the essential uniqueness of the unitary or antiunitary operator
$U$ in~{(\ref{canfor})} is guaranteed by the standard version of Wigner's theorem.
\end{proof}

Still another version of Wigner's theorem can be considered the following result.

\begin{theorem} \label{linani}
Let $\lima$ be a linear operator in $\trcsa$, defined on $\domlima=\frsa$,
and mapping the convex set $\frstah\defi\frsa\cap\stah$ bijectively onto itself.
Then, $\lima$ is closable, and its closure $\limacl$ is a surjective isometry
of the canonical form~{(\ref{canfor})}; namely, a symmetry transformation.
\end{theorem}

\begin{proof}
We first note that
\begin{equation}
\ran(\lima)=\lima(\spanr(\frstah)=\spanr(\lima(\frstah))
=\spanr(\frstah)=\frsa \fin .
\end{equation}
Observe, moreover, that $\lima$ maps $\pru$ into itself. In fact, let
$P$ be a rank-one projection in $\hh$, and let $\jj$ a finite-dimensional subspace
of $\hh$ such that $\ran(\lima(P))\subset\jj$. Clearly, $\lima(P)$ is an extreme
point of the convex set $\staj\subset\frstah$, because $\lima(P)=\alpha\tre S + (1-\alpha)\tre T$,
with $0<\alpha<1$ and $S,T\in\staj$, implies that
\begin{equation}
P = \alpha\sei \lima^{-1}(S) + (1-\alpha)\tre \lima^{-1}(T)\fin ,
\end{equation}
where --- taking into account that $\lima$ maps $\frstah$ bijectively onto itself ---
with a slight abuse we have denoted by $\lima^{-1}(S)$, $\lima^{-1}(T)$
the unique density operators in the sets $\lima^{-1}(\{S\})$ and  $\lima^{-1}(\{T\})$, respectively.
Hence, $\lima^{-1}(S)=\lima^{-1}(T)=P$ and $\lima(P)=S=T$ is a pure state, as claimed.

Since $\lima(\pru)\subset\pru$, as in the proof of Theorem~\ref{linwig} we argue that $\lima$ is positive,
trace-preserving and bounded on its dense domain. Hence, $\lima$ is closable and its closure $\limacl$
is a trace-preserving, positive bounded operator defined on $\dom(\limacl)=\trcsa$. Then, $\limacl$
is a \emph{pure} positive map whose dense range contains $\frsa$ so that --- by Theorem~{3.1} of~\cite{Davies} ---
for every $A\in\trcsa$,  $\limacl(A)= U A \sei U^\ast$, where $U$ is a bounded linear or
antilinear operator in $\hh$. Finally, the fact that $\limacl$ is trace-preserving entails that $U$ is
a (linear or antilinear) isometry: $U^\ast\tre U=\id$; but actually a unitary or antiunitary operator,
because $\ran(\limacl)\supset\frsa$ is dense in $\trcsa$.
\end{proof}

\begin{remark} \label{linkad}
The previous result is closely related to the following fact.
Let $\lima$ be a linear operator --- with $\domlima=\trcsa=\spanr(\stah)$
(every element of $\trcsa$ can be expressed as a linear combination of two density operators) ---
mapping the convex set $\stah$ of all states bijectively onto itself. The restriction of $\lima$ to $\stah$
(which preserves the convex structure) is a Kadison automorphism~\cite{Kadison,Simon},
and then it acts as a symmetry transformation on density operators.
It follows that $\lima$ is bounded and acts in the same way on the whole space $\trcsa=\spanr(\stah)$.
\end{remark}

Recently, it has been proved the following interesting fact~\cite{Sarbicki}, which partially generalizes
Wigner's theorem.

\begin{theorem}[Sarbicki \emph{et al.}] \label{linfin}
For $n=\dim(\hh)<\infty$ prime number, let $\lima\colon\bousa\rightarrow\bousa$,
$\bousa\equiv\trcsa$, be a linear map mapping $\prk$ --- with $k<n$ ---  bijectively onto itself.
Then, $\lima$ is a surjective isometry of the form
\begin{equation} \label{fincanfor}
\lima(A)= U A \sei U^\ast, \ \ \ \forall\cinque A\in\bousa \fin ,
\end{equation}
where $U$ is a unitary or antiunitary operator in $\hh$.
\end{theorem}

Inspired by the previous results, and considering the general case of a possibly
infinite-dimensional Hilbert space, we introduce the following notion.

\begin{definition}[Symmetry witnesses] \label{witness}
A \emph{symmetry witness candidate} is a  set $\syw\subset\trcsa$, invariant with respect to
every unitary or antiunitary transformation, such that its (real) linear span --- $\spanr(\syw)$ ---
is dense in $\trcsa$; i.e., such that $\spancl(\syw)=\trcsa$.
A symmetry witness candidate $\syw$ is called a \emph{symmetry witness} if every injective linear operator
$\lima$ in $\trcsa$ --- defined on $\domlima=\spanr(\syw)$, and mapping $\syw$ onto itself: $\lima(\syw)=\syw$
--- is of the canonical form $\lima(A)= U A \sei U^\ast$, $\forall\cinque A\in\spanr(\syw)$,
for some unitary or antiunitary operator $U$ in $\hh$.
\end{definition}

\begin{remark} \label{remwitness}
Let $\syw$ be a symmetry witness candidate and $\lima$ a linear operator
in $\trcsa$, defined on $\domlima=\spanr(\syw)$, and mapping $\syw$ onto itself.
Then, exploiting the fact that $\lima(\syw)=\syw$, it is straightforward to conclude
that $\lima$ maps $\spanr(\syw)$ onto itself. Therefore, one can equivalently
characterize a symmetry witness $\syw$ as a symmetry witness candidate such that every linear operator
$\lima$ in $\trcsa$, with $\domlima=\spanr(\syw)$, mapping $\spanr(\syw)$ bijectively onto
itself and leaving $\syw$ invariant, $\lima(\syw)=\syw$, has the canonical form
of a unitary or antiunitary transformation.
\end{remark}

\begin{remark}
Given a symmetry witness $\syw$ and an injective linear operator $\lima$, with $\domlima=\spanr(\syw)$,
which leaves $\syw$ invariant, $\lima$ is closable and its closure $\limacl$ is a
surjective isometry of the form~{(\ref{canfor})}; i.e., a symmetry transformation in $\trcsa$.
More generally, given a closed injective linear operator in $\trcsa$ whose domain contains $\spanr(\syw)$,
one can ascertain whether it is a symmetry transformation by checking whether it leaves $\syw$ invariant.
\end{remark}

By Theorem~\ref{linwig}, the set of pure states $\pru$ is a symmetry witness.
By Theorem~\ref{linani} and by Remark~\ref{linkad}, the convex body $\stah$ of all states
and the convex set $\frstah\subset\stah$ of finite-rank density operators are symmetry witnesses as well.
Theorem~\ref{linfin} provides a partial classification of the sets of projections of a fixed rank that are
symmetry witnesses. These examples suggest a further characterization of symmetry witnesses.
We will say that a symmetry witness (candidate) $\syw$ is \emph{positive}, \emph{negative} or \emph{signed}
if it is contained in $\poco$, $\neco$ or $\sico$, respectively. Clearly, if $\syw$ is a \emph{signed}
symmetry witness candidate, then its image $\psyw$ through $\prom$ ---
$\psyw\defi\prom(\syw)\subset\stah$ --- is a symmetry witness candidate too: if $A\in\syw\subset\sico$,
then $U A \sei U^\ast\in\syw$, for any unitary or antiunitary operator in $\hh$, so that
\begin{equation}
U \tre \prom(A) \sei U^\ast = \tr(|A|)^{-1}\tre U \tre |A| \sei U^\ast =
\tr(|U A \sei U^\ast|)^{-1} \tre |U A \sei U^\ast|=\prom(U A \sei U^\ast)\in\psyw \fin .
\end{equation}

\begin{definition}[Projectable symmetry witnesses] \label{project}
We will say that a signed symmetry witness is \emph{projectable} if its image through $\prom$ is a symmetry
witness too.
\end{definition}

\begin{remark} \label{projcase}
Obviously, every positive symmetry witness whose elements have a fixed value of the trace
is projectable. E.g., if $\prk$ is a symmetry witness, then it is projectable; i.e.,
$\stauk$ is a symmetry witness too.
\end{remark}

St{\o}rmer~\cite{Stormer-new} has very recently proved another interesting theorem
of the `Wigner type' which involves unital positive maps acting in the complex
Banach space $\bou$ of bounded operators in $\hh$. Restricting to
selfadjoint operators --- i.e., to the real Banach space $\bousa$ ---
we get the following equivalent formulation.

\begin{theorem}[St{\o}rmer] \label{boucase}
Let $\map\colon\bousa\rightarrow\bousa$ be a unital positive linear map
mapping $\prk$, $k<\dim(\hh)$, onto itself. Then, $\map$ is of the form
\begin{equation}
\map(A)= U A \sei U^\ast, \ \ \ \forall\cinque A\in\bousa \fin ,
\end{equation}
where $U$ is a unitary or antiunitary operator in $\hh$.
\end{theorem}

A comparison with Theorem~\ref{linfin} suggests, however, that, at least in the finite-dimensional
case, the assumptions of positivity and unitality in Theorem~\ref{boucase} may
be dispensed with. We will clarify this point in sect.~\ref{main}.

A further remarkable classical version of Wigner's theorem is due to Ulhorn~\cite{Ulhorn}.
As Wigner's seminal result it concerns maps defined on pure states only,
but it relies on a slightly weaker assumption: only the \emph{vanishing} transition probabilities
are preserved (notice that, given projections $P,Q\in\trcsa$, $P\tre Q=0 \ \Leftrightarrow \ \tr(P\tre Q)=0$).
The following result, due to \v{S}emrl~\cite{Semrl}, is a generalization of Ulhorn's theorem
including maps acting on projections of a fixed rank. Recall that the hypothesis
that the dimension of the relevant Hilbert space is at least three is
a distinguishing feature of Ulhorn's theorem.

\begin{theorem}[\v{S}emrl] \label{teopre}
Let $\lima\colon\prk\rightarrow\prk$ --- with $2\le 2k < \dim(\hh)$ ---
be a bijection that preserves orthogonality in both directions, i.e.,
\begin{equation}
P\tre Q=0 \ \ \Longleftrightarrow \ \ \lima(P)\tre \lima(Q)=0 \fin, \ \ \
P,Q\in\prk \fin .
\end{equation}
Then, $\lima$ is of the form
\begin{equation}
\lima(P)= U P \sei U^\ast, \ \ \ \forall\cinque P\in\prk \fin ,
\end{equation}
where $U$ is a unitary or antiunitary operator in $\hh$.
\end{theorem}

\begin{proof}
The original result is worked out in the case where $\dim(\hh)=\infty$,
but it is clear from the proof that it does hold in the case where $2\le 2k < \dim(\hh)<\infty$,
as well; also see the remark on p.~{572} of~\cite{Semrl}.
\end{proof}

Theorem~\ref{teopre} will turn out to be central in the proof of our main result,
i.e., of Theorem~\ref{mainth} \emph{infra}.

%%%---------------------------------------------------------------------------------
\section{Technical interlude}
\label{interlude}
%%%---------------------------------------------------------------------------------

Before stating and proving the main results of the paper --- see sect.~\ref{main} ---
in this section we need to establish some useful facts.

We start with two simple technical lemmas.

\begin{lemma} \label{lemspa}
For every $k\in\nat$ --- with $k<\dim(\hh)$ --- $\spanr(\prk)=\frsa$.
\end{lemma}

\begin{proof}
It is a well known elementary fact that every rank-one projection
can be written as a (real) linear combination $k+1$ elements of $\prk$,
whenever $k<\dim(\hh)$. The assertion immediately follows.
\end{proof}

\begin{lemma} \label{finilem}
For $n=\dim(\hh)<\infty$, let $\lima\colon\bousa\rightarrow\bousa$ be a linear map
which maps $\prk$ --- with $k<n$ --- onto itself: $\lima(\prk)=\prk$.
Then, $\lima$ maps $\bousa$ bijectively onto itself.
\end{lemma}

\begin{proof}
Since, by Lemma~\ref{lemspa}, $\spanr(\prk)=\frsa=\bousa$ ($n=\dim(\hh)<\infty$) and
$\lima$ maps $\prk$ onto itself, it is easy to check that $\ran(\lima)=\bousa$.
Thus, $\lima$ is a surjective linear map from $\bousa$ onto itself and $\dim(\bousa)=n^2<\infty$,
so that $\dim(\ker(\lima))=0$ and $\lima$ is bijective.
\end{proof}

We are now ready to prove the main facts of this section. The first result is a
generalization of Proposition~{1} of~\cite{Sarbicki}, and takes into account, in particular,
the fact that the Hilbert space we deal with may be infinite-dimensional. Next, a simple observation
--- Proposition~\ref{uncase} below --- will allow us to obtain a suitable refinement of this result,
refinement that will be central in the proof of the main theorem.

\begin{proposition} \label{lemprea}
Let $\lima$ be a linear operator in $\trcsa$, defined on $\dom(\lima)=\frsa$.
Suppose that for some $k\in\nat$ --- with $2k\le \dim(\hh)$, if $\dim(\hh)<\infty$ ---
$\lima$ maps $\prk$ onto itself. In the case where $\dim(\hh)=\infty$ and $k>1$, assume also that
$\lima$ is injective. Then, $\lima$ preserves the mutual orthogonality in $\prk$:
\begin{equation} \label{dir}
P,Q\in\prk \fin , \ P\tre Q=0  \ \ \Longrightarrow \ \ \lima(P)\quattro \lima(Q)=0 \fin .
\end{equation}
\end{proposition}

\begin{proof}
By Theorem~\ref{linwig}, relation~{(\ref{dir})} is true for $k=1$ (note, however, that the
following proof works for $\dim(\hh)<\infty$ and $k=1$, as well).

The statement is proven if we show that $\lima$ maps $\prdk$ into itself.
Indeed, in this case, we have:
\begin{eqnarray}
P,Q\in\prk, \  P\tre Q=0  & \Longrightarrow &  \lima(P+Q)= \lima(P) + \lima(Q)\in\prdk,
\ \lima(P),\lima(Q)\in\prk
\nonumber \\
& \Longrightarrow &  \lima(P)\tre\lima(Q)=0 \fin .
\end{eqnarray}

Then, let $\pkk$ be the orthogonal projection onto a $2k$-dimensional subspace $\kk$ of $\hh$.
Let us express $\pkk$ as the sum of two --- necessarily mutually orthogonal ---
projections of rank $k$: $\pkk=P+Q$, $P,Q\in\prk$. Setting $\imp\equiv\lima(P)\in\prk$,
$\imq\equiv\lima(Q)\in\prk$ and $\impkk\equiv\lima(\pkk)$, it turns out that
\begin{equation}
\ran(\imp), \tre \ran(\imq) \subset \ran (\impkk) \fin .
\end{equation}
In fact, we have that
\begin{eqnarray}
\ran(\imp)= \ker(\imp)^\perp \spa & \subset & \spa
\big(\mcinque\ker(\imp)\cap\tre\ker(\imq)\big)^\perp
\nonumber \\
& = & \spa
\big(\ran(\imp)^\perp\mtre\cap\tre\ran(\imq)^\perp\big)^\perp
\nonumber \\
& = & \spa
\ran(\imp) + \ran(\imq)
\nonumber \\
& = & \spa \ran(\imp+\imq) = \ran(\impkk) \fin ,
\end{eqnarray}
and analogously $\ran(\imq)\subset\ran (\impkk)$. Hence, denoting by $\prkk\subset\prk$
the set of projections of rank $k$ with range contained in $\kk$, we conclude that
\begin{equation}
\ran(\lima(P))\subset\imkk\equiv\ran(\impkk) \fin , \ \ \ \forall\cinque P \in\prkk \fin .
\end{equation}
Denoting by $\bkk$, $\bimkk$, the spaces of all (finite rank) linear selfadjoint operators in $\hh$
with range contained in $\kk$ and $\imkk$, respectively, by the previous relation we have that
\begin{equation} \label{cont}
\imprkk\defi\lima(\prkk)\subset\primkk\subset\bimkk
\end{equation}
and, because $\spanr(\prkk)=\bkk$,
\begin{equation} \label{inclu}
\imbkk\defi\lima(\bkk)=\lima(\spanr(\prkk))
=\spanr(\lima(\prkk))=\spanr(\imprkk)\subset\bimkk \fin .
\end{equation}

By the previous relation, we can define the linear map
\begin{equation}
\limr\colon\bkk\ni A\mapsto\lima(A)\in\bimkk \fin .
\end{equation}
Observe that $\limr$ is \emph{injective}: if $\dim(\hh)<\infty$, by the fact that
$\lima$ maps $\prk$ onto itself and by Lemma~\ref{finilem}, so that $\lima$ is injective;
if $\dim(\hh)=\infty$ and $k>1$, because in this case $\lima$ is injective by assumption.

Since $\limr$ is an injective linear map, the dimension of its domain coincides
with that of its range: $\dim(\bkk)=\dim(\imbkk)$; hence:
\begin{equation}
\dim(\bkk)=\dim(\imbkk)\le\dim(\bimkk) \ \ \Longrightarrow \ \
\dim(\kk)\le\dim(\imkk) \fin .
\end{equation}
On the other hand, we also have:
\begin{equation}
\impkk = \imp + \imq \ \ \Longrightarrow \ \
\dim(\imkk)=\dim(\ran(\impkk))\le 2k = \dim(\kk) \fin .
\end{equation}
Hence, it is now clear that
\begin{equation}
2k=\dim(\kk)=\dim(\imkk)=\dim(\ran(\impkk))=\rank(\impkk) \fin ,
\ \ 4k^2=\dim(\bkk)=\dim(\bimkk) \fin .
\end{equation}
Moreover, $\limr$ is a bijection and
\begin{equation}
\lima(\bkk) = \limr(\bkk) = \bimkk \fin .
\end{equation}
Therefore, $\lima$ maps $\bkk$ bijectively onto $\bimkk$ and,
in particular, it maps $\prkk$ bijectively onto $\primkk$
(note that $\lima^{-1}(\primkk)=\bkk\cap\prk=\prkk$, because $\lima(\prk)=\prk$ and
$\lima$ is injective).

Now, take any $\neop\in\primkk$ and observe that
\begin{equation}
\impkk - \neop = \lima(\pkk) - \lima(P) = \lima(\pkk - P) \fin ,
\ \ \ \mbox{for some} \ P\in\prkk \fin .
\end{equation}
Thus: $\neoq\defi\impkk - \neop\in\primkk$. One can choose $\neop$ in such a way that
\begin{equation} \label{commu}
\neop\sei\impkk=\impkk\sei\neop \fin ;
\end{equation}
this amounts to choosing $\neop$ as the sum of mutually orthogonal rank-one eigenprojections of $\impkk$.
By the definition of $\neoq$ and by~{(\ref{commu})}, we have:
\begin{equation}
\neoq\tre\neop=\big(\impkk - \neop\big) \neop = \impkk \tre\neop - \neop, \ \ \
\neop\tre\neoq = \neop\tre\impkk - \neop = \impkk \neop - \neop
= \neoq\tre\neop \fin .
\end{equation}
It follows that $\impkk$ can be expressed as
\begin{equation} \label{decom}
\impkk = \bp + \bq + 2R \fin ,
\end{equation}
where
\begin{equation}
\bp\defi \neop - \neop\tre\neoq \fin , \ \
\bq\defi \neoq - \neop\tre\neoq \ \
\mbox{and} \ \ R\defi\neop\tre\neoq = \neoq\tre\neop
\end{equation}
are \emph{mutually orthogonal projections} (as the reader may easily check);
in particular, $R$ is the projection onto $\ran(\neop)\cap\tre\ran(\neoq)$.
Observe that from the decomposition~{(\ref{decom})} we derive the following conclusion
about the rank of $\impkk$ (one may equivalently argue using the trace):
\begin{equation}
2k=\rank(\impkk)=\rank(\bp) + \rank(\bq) + \rank(R) =
(k-\rank(R)) + (k-\rank(R)) + \rank(R) \fin .
\end{equation}
Hence, $R=0$ and the projections $\neop$, $\neoq$ are mutually orthogonal, so that
$\impkk=\lima(\pkk)=\pimkk$ is a projection too, the projection onto the subspace $\imkk$.

By the arbitrariness of the $2k$-dimensional subspace $\kk$, it follows
$\lima$ maps $\prdk$ into itself, and the proof is complete.
\end{proof}

\begin{proposition} \label{uncase}
Let $\hh$ be infinite-dimensional, and let $\lima$ be a linear operator in $\trcsa$,
defined on $\dom(\lima)=\frsa$, and mapping $\prk$ onto itself. If $k>1$, assume also that
$\lima$ is injective. Then, $\lima$ maps $\frsa$ bijectively onto itself.
\end{proposition}

\begin{proof}
If $k=1$, the statement is true by Theorem~\ref{linwig}.
If ($\dim(\hh)=\infty$ and) $k>1$, $\lima$ is injective by hypothesis,
and, exploiting Lemma~\ref{lemspa}, one concludes that $\lima(\frsa)=\frsa$, as well.
\end{proof}

As anticipated, exploiting the previous fact we can now obtain a very useful refinement
of Proposition~\ref{lemprea}.

\begin{proposition} \label{probodir}
With the same assumptions of Proposition~\ref{lemprea}, $\lima$ maps
$\frsa$ and $\prk$ bijectively onto themselves, and preserves the mutual orthogonality of
elements of $\prk$ in \emph{both} directions; i.e.,
\begin{equation} \label{bodir}
P\tre Q=0 \ \ \Longleftrightarrow \ \ \lima(P)\tre \lima(Q)=0 \fin, \ \ \
P,Q\in\prk \fin .
\end{equation}
\end{proposition}

\begin{proof}
With the assumptions of Proposition~\ref{lemprea}, it turns out that
$\lima$ maps $\frsa$ bijectively onto itself. If $\dim(\hh)<\infty$, by the fact that
$\lima$ maps $\prk$ onto itself and by Lemma~\ref{finilem}.
If $\dim(\hh)=\infty$, by Proposition~\ref{uncase}.
In particular, $\lima$ is injective and $\lima(\prk)=\prk$; i.e., it maps $\prk$
bijectively onto itself.

Then, to complete the proof just observe that the inverse operator $\lima^{-1}$,
which is defined on $\dom(\lima^{-1})=\ran(\lima)=\frsa$, of course
verifies the assumptions of Proposition~\ref{lemprea} too. Hence,
relation~{(\ref{dir})} actually holds in \emph{both} directions.
\end{proof}

The following lemma is essentially Proposition~{2} of~\cite{Sarbicki} (we include
the proof here for the reader's convenience).

\begin{lemma} \label{teclem}
Suppose that $n=\dim(\hh)<\infty$ and set $m=n-k$, for some $k\in\nat$, with $k<n$.
If every linear map $\limb\colon\bousa\rightarrow\bousa$,
mapping $\prm$ bijectively onto itself, is of the canonical form~{(\ref{fincanfor})},
then every linear map $\lima\colon\bousa\rightarrow\bousa$, mapping $\prk$ bijectively
onto itself, is of that canonical form too.
\end{lemma}

\begin{proof}
Let $\lima\colon\bousa\rightarrow\bousa$ be linear, and let it map $\prk$ bijectively onto itself,
for some $k<n=\dim(\hh)$; set $m=n-k$. For any $P\in\prk$, $Q=\pper=\id-P$ belongs to $\prm$ and
\begin{equation} \label{firel}
\lima(\id)-\lima(Q)=\lima(P)\ifed\neop\in\prk \fin .
\end{equation}
Consider the projection $\neoq\defi\neopper= \id -\neop\in\prm$; clearly, by~{(\ref{firel})},
\begin{equation} \label{strt}
\neoq = \lima(Q) - \lima(\id) + \id  \fin .
\end{equation}
Therefore, one can define the map
\begin{equation}
\prelima\colon\prm\ni Q\mapsto \lima(Q) - \lima(\id) + \id \in\prm \fin ,
\end{equation}
which is a bijection. In fact, it is easy to see that its inverse is given by
\begin{equation} \label{forpreli}
\prelimmo(\neoq)=(\lima^{-1}(\neoqper))^\perp,
\end{equation}
where we have used the fact that, by Lemma~\ref{finilem}, $\lima$ is bijective
(this is not essential, however; we could have replaced $\lima^{-1}(\neoqper)$,
in~{(\ref{forpreli})}, with $P$ --- $\{P\}=\lima^{-1}(\{\neoqper\})$ --- since $\lima$ is
injective on $\prk$).

Let us now `embed' the map $\prelima$ in the linear map
\begin{equation}
\limb\colon\bousa\ni A \mapsto \lima(A) -
m^{-1} \tre \tr(A)\tre (\lima(\id)-\id)\in\bousa \fin .
\end{equation}
Observe that $\limb$ acts precisely like $\prelima$ on $\prm$; hence,
it maps $\prm$ bijectively onto itself. Therefore, by hypothesis,
$\limb$ must be of the form
\begin{equation}
\limb(A)= U A \sei U^\ast, \ \ \ \forall\cinque A\in\bousa \fin ,
\end{equation}
where $U$ is a unitary or antiunitary operator in $\hh$.
It follows that
\begin{equation}
\lima(A)= U A \sei U^\ast + m^{-1} \tre \tr(A)\tre (\lima(\id)-\id) \fin .
\end{equation}
Let us show that the second term on the rhs of this expression vanishes;
i.e., that $\lima(\id)-\id=0$.

Indeed, in particular, for $A=P\in\prk$, we have:
\begin{equation}
\prk\ni\neop\defi\lima(P)= \imp + R \fin , \ \ \
\imp\defi U P \sei U^\ast\in\prk \fin , \ \ \
R\defi k \sei m^{-1} \tre (\lima(\id)-\id) \fin .
\end{equation}
The fact $\neop$ is a projection entails that
\begin{equation}
\imp + R=\neop = \neopsq = \imp + R^2 + \imp R + R\tre\imp \
\ \Longrightarrow \ \ \imp R + R\tre\imp = R - R^2 \fin .
\end{equation}
Hence, taking into account that $\tr(R)=\tr(\neop)-\tr(\imp)=k - k=0$,
we find that
\begin{equation} \label{nonpos}
2\sei \tr(\imp R)= \tr(R) - \tr\big(R^2\big)=
- \tr\big(R^2\big)\le 0 \fin .
\end{equation}
It is straightforward to check that the previous inequality, which holds for every $\imp\in\prk$,
together with the fact that $\tr(R)=0$, implies that $R=0$. To this aim, choose $\imp$
in~{(\ref{nonpos})} as the sum of mutually orthogonal rank-one eigenprojections of $R$ associated
with the eigenvalues $\lambda_1,\ldots,\lambda_k$, having denoted by $\lambda_1\ge\lambda_2\ge\cdots$
the eigenvalues of $R$, in decreasing order and repeated according to degeneracy.
\end{proof}

%%%---------------------------------------------------------------------------------
\section{Main results}
\label{main}
%%%---------------------------------------------------------------------------------

The sets $\prk$, $k<\dim(\hh)$, are symmetry witness \emph{candidates} (in particular,
their linear span is dense in $\trcsa$; see Lemma~\ref{lemspa}). Now, equipped with
the results outlined in sect.~\ref{background} and with the technical facts derived
in the previous section, we will be able to find out in what cases they actually
\emph{are} symmetry witnesses.

The linear version of Wigner's theorem --- Theorem~\ref{linwig} ---
and the result by Sarbicki \emph{et al.}~\cite{Sarbicki} ---
i.e., Theorem~\ref{linfin} --- turn out to be special cases of our main theorem.

\begin{theorem} \label{mainth}
Let $\hh$ be a sparable complex Hilbert space, and let $\lima$ be a linear operator in $\trcsa$,
defined on the dense domain $\dom(\lima)=\frsa$. Suppose that given some $k\in\nat$
--- with $k=1$, if $\dim(\hh)=2$; with $k<n/2$ or $n/2<k<n$, if $3\le n=\dim(\hh)< \infty$;
arbitrary, if $\dim(\hh)=\infty$ ---
the following conditions are satisfied:
\begin{enumerate}

\item $\lima(\prk)=\prk$;

\item $\lima$ is injective.

\end{enumerate}
(If $\dim(\hh)<\infty$ --- or $\dim(\hh)=\infty$ and $k=1$ --- the second hypothesis
is superfluous, because it is actually implied by the first.)

Then, $\lima$ maps $\frsa$ bijectively onto itself. Moreover, it is closable,
and its closure $\limacl$ is a surjective isometry of the form
\begin{equation} \label{caform}
\limacl(A) = U A \sei U^\ast, \ \ \ \forall\cinque A\in\trcsa \fin ,
\end{equation}
where $U$ is a unitary or antiunitary operator in $\hh$, uniquely defined up to
multiplication by a phase factor.
\end{theorem}

\begin{proof}
Observe that, according to Lemma~\ref{finilem} and to Theorem~\ref{linwig}, respectively, we have:
if $\hh$ is finite-dimensional, the hypothesis that $\lima$ maps
$\prk$ onto itself --- $k<\dim(\hh)$ --- implies that $\lima$ is injective; if $\dim(\hh)=\infty$ and $k=1$,
the assumption that $\lima$ maps $\pru$ onto itself implies, once again,
that $\lima$ is injective. Thus, the above assertion between parentheses is clarified.

Let us now consider first the case where $k=1$. If $\lima(\pru)=\pru$,
then, by Theorem~\ref{linwig}, $\lima$ is closable and its closure $\limacl$
is of the form~{(\ref{caform})}.

Let us next suppose that $k\ge 2$ and, at first, that $2k<\dim(\hh)$ (thus, $\dim(\hh)\ge 5$),
so that the assumptions of both Proposition~\ref{probodir} --- that coincide
with those of Proposition~\ref{lemprea} --- and Theorem~\ref{teopre} are satisfied.
By the first result, the operator $\lima$, which maps $\prk$ bijectively onto itself,
preserves the mutual orthogonality of the elements of $\prk$, in both directions.
Then, by Theorem~\ref{teopre}, for every $P\in\prk$, we have that $\lima(P)= U P \sei U^\ast$,
for some unitary or antiunitary operator $U$ in $\hh$. By Lemma~\ref{lemspa},
this form of the operator $\lima$ extends by linearity to its domain
$\dom(\lima)=\frsa=\spanr(\prk)$. Then, by the density of $\frsa$ in $\trcsa$, $\lima$
extends to the surjective isometry $\limacl$ defined by~{(\ref{caform})},
which is the closure of $\lima$.

It remains to consider the case where $3\le n=\dim(\hh)<\infty$ and
$n/2<k<n$ ($k\ge 2$). Notice that, in this case, using Lemma~\ref{teclem}
--- with $1\le m=n-k<n/2$ --- one can reduce the problem to one of the previously considered
cases where $1\le k<n/2=\dim(\hh)<\infty$.

Finally, the fact that the unitary or antiunitary operator $U$ in~{(\ref{caform})}
is uniquely defined up to a phase factor follows from Wigner's theorem.
\end{proof}

\begin{remark}
Note that in the preceding proof we have used two different techniques. For $k=1$, we have referred to
Theorem~\ref{linwig}, whose proof relies on a classical result of Davies~\cite{Davies-book,Davies}. For $k>1$, instead,
we have exploited a generalization of Ulhorn's theorem; i.e., Theorem~\ref{teopre}. However, in the
finite-dimensional setting we may have adopted the latter technique for $k=1$ too,
with the only exception of the case where $\dim(\hh)=2$.
This, in fact, is the well known peculiar exception to Ulhorn's theorem.
\end{remark}

\begin{remark}
Clearly, for $3\le n=\dim(\hh)<\infty$, the constraints on $k$ in Theorem~\ref{mainth} entail that,
for $n$ odd, every value of $k<n$ allows one to apply our result, whereas, for $n\ge 4$ and even, the value $k=n/2$
is not contemplated.

Actually, the following counter-example shows that the conclusion of the theorem does not hold in that
particular case. Consider the linear map $\linp\colon\bousa\equiv\trcsa\rightarrow\bousa$ --- with $2\le n=\dim(\hh)=2k$,
$k\in\nat$ --- defined by
\begin{equation} \label{deflinp}
\linp(A)\defi k^{-1} \tre \tr(A)\otto \id - A \fin .
\end{equation}
This map is trace-preserving and unital, but, for $k>1$, \emph{not} positive.
Notice that, for $k>1$, it maps every pure state to a non-positive operator of
rank $n-1=2k-1\neq 1$. $\linp$ maps $\prk$ bijectively onto itself, but,
for $k>1$, it is \emph{not} a symmetry transformation.\footnote{For $k=1$,
we get the celebrated \emph{reduction map} $\linpone$, which, represented in terms of $2\times 2$
hermitian matrices associated with an orthonormal basis, is the form
$M\mapsto\siy\quattro M^\trasp \siy$; thus, a symmetry transformation,
coherently with Theorem~\ref{mainth}.}
It is also worth observing that
\begin{equation} \label{relinp}
\linp\circ\linp = \ide \ \ \ \mbox{and} \ \ \
\linp (U A \sei U^\ast) = U \tre \linp(A) \tre U^\ast,
\end{equation}
where $U$ is a unitary or antiunitary operator in $\hh$. Thus, one can define the
map
\begin{equation} \label{deflinpu}
\linpu\colon \bousa\ni A \mapsto \linp (U A \sei U^\ast) \in \bousa \fin ,
\end{equation}
which, once again, is trace-preserving and unital, but, for $k>1$, \emph{not} positive.
Observe, moreover, that by relations~{(\ref{relinp})} the map $\linpu\circ\linpu$ is always
a unitary transformation.
\end{remark}

\begin{remark}
Let us notice explicitly that Theorem~\ref{mainth} entails that
the symmetry witness candidate $\prk$
--- with the previously specified constraints on $k$ --- is actually a symmetry witness.
In fact, if a linear operator $\lima$ in $\trcsa$, densely defined on
$\domlima=\spanr(\prk)=\frsa$, maps $\prk$ onto itself and
is injective, then it acts as a unitary or antiunitary transformation.
\end{remark}

From Theorem~\ref{mainth} and the previous two remarks, we derive the following fact.

\begin{corollary} \label{coro}
Theorem~\ref{mainth} provides a complete classification of the sets of projections
of a fixed rank that are symmetry witnesses. In particular,
in the case where $4\le n=\dim(\hh)<\infty$ and $n$ is \emph{even},
a linear operator in $\bousa$ mapping $\prnm$ onto itself
is a trace-preserving and unital bijection, but not, in general, a symmetry witness.
\end{corollary}

\begin{proof}
The only assertion to be proven is the one concerning $\prnm$. Indeed, if
a linear map $\lima\colon\bousa\rightarrow\bousa$, with $4\le n=\dim(\hh)<\infty$ and $n$ even,
maps $\prnm$ onto itself, it is bijective (by Lemma~\ref{finilem}) and trace-preserving. Moreover,
by Proposition~\ref{probodir}, expressing the identity in $\bousa$ as the sum
of two mutually orthogonal projections of rank $n/2$, it is clear that $\lima(\id)=\id$, because
the mutual orthogonality of these projections is preserved by $\lima$.
\end{proof}

Recalling Remark~\ref{projcase}, observe that the sets of projections of a fixed rank
that are symmetry witnesses are, in particular, \emph{projectable} positive symmetry witnesses.

\begin{corollary}
Let $\lima$ be a linear operator in $\trcsa$, defined on $\dom(\lima)=\frsa$
and mapping the set $\stauk=\prom(\prk)$, of uniform density operators of a fixed rank $k<\dim(\hh)$,
onto itself. Suppose, moreover, that
\begin{itemize}

\item in the case where $\dim(\hh)=\infty$ and $k>1$, $\lima$ is injective;

\item in the case where $2<2k=\dim(\hh)<\infty$, $\lima$ is positive.

\end{itemize}
Then, the operator $\lima$ is closable, and its closure $\limacl$
is a surjective isometry of the form~{(\ref{caform})}.
If $k<\dim(\hh)$ and --- for $\dim(\hh)\ge 4$ ---
$2k\neq\dim(\hh)$, $\stauk$ is a symmetry witness; otherwise, it is not.
\end{corollary}

\begin{proof}
Taking into account Theorem~\ref{mainth}, Corollary~\ref{coro} and Remark~\ref{projcase},
the only case that we still need to prove is $2k=\dim(\hh)<\infty$, $k>1$, with $\lima$
assumed to be positive. By Corollary~\ref{coro}, we know that $\lima$ is unital, as well.
Hence, by Theorem~\ref{boucase}, $\lima$ is a unitary or antiunitary transformation.
\end{proof}

To conclude this section, it is interesting to consider how Theorem~\ref{mainth}
would be expressed if we had chosen to work in the \emph{complex} Banach space $\trc$
of all trace class operators, rather than in the real Banach space $\trcsa$.

Then, let $\limco$ be a linear operator in $\trc$, with $\domlima=\fr$ (the linear manifold
of finite rank operators in $\trc$), and --- for some $k\in\nat$ satisfying the constraints
specified in the statement of Theorem~\ref{mainth} --- let $\limco$ map $\prk$
onto itself. If $\hh$ is infinite-dimensional and $k>1$, let us further assume that $\limco$
be injective. Since $\spanr(\prk)=\frsa=\fr\cap\trcsa$, the operator $\limco$ is adjoint-preserving
--- $\limco(F^\ast) = \limco(F)^\ast$ --- and $\limco(\frsa)\subset\frsa$. Hence,
$\limco$ induces a linear operator $\lima$ in the real Banach space $\trcsa$, with
$\domlima=\frsa$ --- i.e., $\lima(A)\defi\limco(A)$, $A\in\frsa$ ---
that satisfies the assumptions of Theorem~\ref{mainth}.
It follows that $\lima$ is closable, and its closure $\limacl$ is a unitary or
antiunitary transformation in $\trcsa$. Besides, $\lima$ determines $\limco$
completely; i.e.,
\begin{equation} \label{complexi}
\limco(A + \ima\quattro B) = \lima(A) + \ima\sei \lima(B) \fin ,
\ \ \ A,B\in\trcsa \fin .
\end{equation}
Therefore, $\limco$ is closable too, and its closure $\limcocl$ is a surjective isometry
in $\trc$, which we will now describe. By~{(\ref{complexi})}, it is clear that $\limcocl$
is either of the form
\begin{equation} \label{compforma}
\limcocl(C) = U \tre C \sei U^\ast, \ \ \ \forall\cinque C\in\trc \fin ,
\end{equation}
for some unitary operator $U$, or of the form
\begin{equation} \label{compformb}
\limcocl(C) = U \tre C^\ast\tre U^\ast, \ \ \ \forall\cinque C\in\trc \fin ,
\end{equation}
for some \emph{anti}unitary operator $U$. Here, it is clear that we may have
considered, equivalently, the transformations generated by unitary operators only,
possibly composed with a transposition associated with an orthonormal basis.
Interestingly, for $\dim(\hh)=2k$, $k\in\nat$, considering the natural extension
to $\bou\equiv\trc$ of the linear map $\linp$ defined by~{(\ref{deflinp})}
(extension which we will still denote by the same symbol), one can again
define as in~{(\ref{deflinp})} the related map $\linpu$, for $U$ unitary; whereas
--- if $U$ is \emph{anti}unitary, instead --- one has to set
\begin{equation} \label{deflinpuco}
\linpu(C) \defi k^{-1} \tre \tr(C)\otto \id -  U \tre C^\ast\tre U^\ast
= \linp(U \tre C^\ast\tre U^\ast) = U \tre (\linp(C))^\ast\tre U^\ast \fin .
\end{equation}
Here we have used the fact that, for $U$ antiunitary, $\tr(U \tre C^\ast \tre U^\ast)=\tr(C)$.
Note that, as in the real setting, $\linpu\circ\linpu$ is always a unitary transformation.

%%%--------------------------------------------------------------------------------
\section{Final remarks and conclusions}
\label{conclusions}
%%%--------------------------------------------------------------------------------

A symmetry witness is a suitable subset $\syw$ of the space $\trcsa$ of selfadjoint trace class
operators that allows one to determine whether an injective linear operator $\lima$ in $\trcsa$, whose domain
contains $\syw$, is a symmetry transformation, depending on whether it leaves the witness invariant;
i.e., whether $\lima(\syw)=\syw$. We observe explicitly that it is not the detailed action of the operator
which is relevant here: the only relevant information is whether or not $\syw$ is invariant with respect to $\lima$.

By a version of Wigner's theorem --- see Theorem~\ref{linwig} --- the set of pure states $\pru$ is the
\emph{prototype} of a symmetry witness. The fact that dealing with \emph{linear} maps makes the
assumption of preservation of transition probabilities superfluous is known, 
at least in the finite-dimensional case~\cite{Grabowski,Sarbicki}.
Further examples of symmetry witnesses are the convex body $\stah$ of all states (Remark~\ref{linkad})
and the convex set $\frstah\subset\stah$ of finite-rank density operators (Theorem~\ref{linani}).

By a natural generalization of Theorem~\ref{linwig} --- i.e., by our main result, Theorem~\ref{mainth} --- the set
$\prk$, $k<\dim(\hh)$, of projections of a (finite) fixed rank is a symmetry witness too,
with the only exception --- in the case where $4\le n=\dim(\hh)<\infty$ and with $n$ \emph{even} ---
of the value $k=n/2$. In this particular case, a linear operator $\lima$ mapping $\prnm$ bijectively onto
itself is a trace-preserving and unital linear bijection, but, in general, \emph{not} a symmetry transformation.
One obtains a symmetry transformation by further assuming that $\lima$ be a \emph{positive} map.
Therefore, $\prnm$ may be regarded as a \emph{conditional} symmetry witness.

The previously classified symmetry witnesses are \emph{positive} (i.e., consist of positive operators)
and \emph{projectable}; namely, they give rise, via the map $\prom\colon\sico\rightarrow\stah$, to
the symmetry witnesses $\stauk$ that have a direct physical interpretation. They are formed by those
density operators of a (finite) fixed rank whose eigenvalues form a \emph{uniform} (i.e., constant)
probability distribution. Note, moreover, that the set $\stauk$ has a precise meaning in relation with
the natural \emph{majorization} preorder in $\stah$ (see~\cite{Aniello-QE,Aniello-SchCo}, and references therein).
It consists of those elements that are minimal, with respect to majorization, among all elements of
$\stah$ of rank not larger than $k$.

The mentioned results provide a complete solution of a problem considered in~\cite{Sarbicki}:
to achieve a generalization of Wigner's theorem in terms of projections of a fixed rank.
The problem is solved in~\cite{Sarbicki} in the finite-dimensional setting, with the assumption
$n=\dim(\hh)$ be a prime number; see Theorem~\ref{linfin}. A problem of this kind has been
also considered in~\cite{Stormer-new} --- see Theorem~\ref{boucase} --- in a different setting.
Namely, the sets of a (finite) fixed rank projections are embedded in the Banach space of bounded
operators; i.e., in the \emph{dual} of the Banach space of trace class operators. Moreover,
the linear operators are assumed to be positive and unital in~\cite{Stormer-new}. Clearly,
a \emph{direct} comparison with the results of~\cite{Sarbicki}, and with those obtained in
the present contribution, is possible in the finite-dimensional setting only. It turns out that,
in this case, both positivity and unitality can be dispensed with, with the only exception of the previously
mentioned case of $\prnm$, for $n=\dim(\hh)\ge 4$ and even. For these values of $n$, unitality is automatically satisfied
by a linear map mapping $\prnm$ onto itself, whereas positivity entails that the map is
a symmetry transformation. In this regard, we stress that considering the real Banach space $\trcsa$
versus the complex Banach space $\trc$ as a general setting is immaterial; see the discussion
concluding sect.~\ref{main}.

As a final remark, we observe that, according to an interesting conjecture~\cite{Sarbicki},
in the case where $4\le n=\dim(\hh)<\infty$ and $n$ is \emph{even}, every linear map
in $\bousa$ mapping $\prnm$ onto itself would be either a symmetry transformation or
a (trace-preserving, unital but) non-positive map of the form $\linpu$; see~{(\ref{deflinpu})}.
It is not clear, at the moment, whether our results can shed new light on this
intriguing problem too. Certainly, the conjecture can now be reconsidered in a somewhat clearer
general context.

%%%--------------------------------------------------------------------------------------------------------

\section*{Acknowledgments}

D.C.\ was partially supported by the National Science Center project 2015/19/B/ST1/03095.

%%%---------------------------------------------------------------------------------------------------------

%%%------------------------------------------------------------------------------------------

\end{document}